\documentclass{article}

\usepackage[english]{babel}

\usepackage[authoryear]{natbib}

\usepackage{amsmath,amssymb,amsfonts,mathtools}
\usepackage{bm}
\usepackage{microtype}

\usepackage[letterpaper,top=2cm,bottom=2cm,left=3cm,right=3cm,marginparwidth=1.75cm]{geometry}

\usepackage{amsmath, amssymb, amsthm}
\usepackage{graphicx}
\usepackage[colorlinks=true, allcolors=blue]{hyperref}

\usepackage[letterpaper,top=2cm,bottom=2cm,left=3cm,right=3cm,marginparwidth=1.75cm]{geometry}
\usepackage{tikz}
\usetikzlibrary{arrows,positioning}
\definecolor{boxcolor}{HTML}{B9DCFF}
\usetikzlibrary{arrows.meta}
\usepackage{subcaption}
\usepackage{float}
\usepackage{multirow}
\usepackage{booktabs}
\usepackage{pgfplots}
\pgfplotsset{compat=1.18} 
\usepackage{amsfonts}
\usepackage{amsmath, graphicx, epsfig,psfrag, verbatim}
\usepackage{eucal,  amssymb,amsmath,amsthm,graphicx, amsfonts, latexsym,xcolor}
\usepackage[utf8]{inputenc}

\usepackage{amsmath}
\usepackage{graphicx}
\usepackage[colorlinks=true, allcolors=blue]{hyperref}
\newtheorem{theorem}{Theorem}
\newtheorem{cor}{Corollary}
\newtheorem{rek}{Remark}

\newtheorem{prop}{Proposition}

\title{Dependence Structure and Epidemic Outcomes in Heterogeneous SIR Models}
\author{Mohamed El Khalifi\\
\small Department of Mathematics, Dhar El Mahraz Faculty of Science, Sidi Mohamed Ben Abdellah University\\
\small Atlas Fez, BP 1796, Fez, Morocco\\
\small \texttt{mohamed.elkhalifi1@usmba.ac.ma}}

\begin{document}
\maketitle

\begin{abstract}
We study a well-mixed SIR epidemic model with heterogeneous susceptibility and infectivity, allowing for an arbitrary joint distribution of these traits. Using an exact final-size formulation and a branching–process approximation for early epidemic dynamics, we show that both the final epidemic size and the probability of a major outbreak are monotone with respect to the concordance order of the joint susceptibility–infectivity distribution. In particular, among all couplings with fixed marginal trait distributions, comonotonic dependence maximizes epidemic severity, yielding sharp distribution-free upper bounds. A key implication is that epidemic outcomes cannot be ordered by susceptibility heterogeneity alone: while increasing susceptibility variance reduces the final size under independence, positive dependence between susceptibility and infectivity can locally increase epidemic size for any basic reproduction number exceeding one. We further show that neither susceptibility variance nor the sign of the covariance suffices to determine epidemic severity under dependence. These findings offer new insights into epidemic risk assessment under limited information.
\end{abstract}
\section{Introduction}
Host heterogeneity is now understood to be a fundamental driver of infectious disease dynamics. Variability in individual susceptibility to infection and infectivity upon infection arises from biological factors (e.g. pathogen load), behavioral heterogeneity (e.g. contact structure), and social structure can profoundly alter epidemic growth and final epidemic size \citep{lloyd2005superspreading,britton2020mathematical}. During but also prior to the Covid-19 pandemic, these insights have motivated a large literature extending the classical homogeneous SIR framework to incorporate individual-level variation, most commonly through random susceptibility, random infectivity, or both \citep{tkachenko2021time,gomes2022individual}.

A robust conclusion emerging from early work on heterogeneous susceptibility is that, under independence or one-sided heterogeneity, increasing variance in susceptibility tends to reduce the epidemic attack rate relative to a homogeneous population with the same mean susceptibility. This phenomenon, often attributed to selective depletion of highly susceptible individuals, has been demonstrated by many scholars including, e.g. \cite{ball1993final,dwyer1997host,novozhilov2008spread}, and \cite{gomes2022individual}, among others. 

However, much of this literature relies—often implicitly—on restrictive assumptions about the relationship between susceptibility and infectivity. In particular, susceptibility and infectivity are frequently assumed to be independent, perfectly correlated, or linked by a fixed parametric rule. In realistic populations, such assumptions are difficult to justify. Individuals who are more likely to become infected (through higher exposure, biological predisposition, or behavior) are often also more likely to transmit once infected, inducing positive dependence between susceptibility and infectivity. Conversely, behavioral responses, partial immunity, or targeted interventions may generate negative dependence. Understanding epidemic outcomes therefore requires separating the effects of marginal heterogeneity from those of the dependence structure between traits.

Recent work has begun to explore this distinction.
\cite{kawagoe2021epidemic} introduced a well-mixed SIR model with distributed susceptibility and infectivity and showed that correlations between these traits can substantially modify early growth rates and final epidemic sizes. Their analysis establishes final-size comparisons via Laplace-type sufficient conditions tailored to the epidemic functional, rather than through an explicit ordering of joint trait distributions.
Here, we instead employ the concordance order, which provides a clear distribution-level ordering of susceptibility–infectivity dependence and yields distribution-free monotonicity results for both the final epidemic size and the major outbreak probability.
Numerical evidence provided by \cite{tuschhoff2025heterogeneity} further demonstrated that positive dependence between susceptibility and infectivity can lead to larger, faster, and more likely epidemics, including major outbreaks in regimes where commonly used reproduction numbers suggest subcriticality. The present work provides an analytical framework that explains and rigorously establishes these observations, thereby offering a complementary theoretical understanding of the effects of joint trait dependence.

In this paper, we study a well-mixed SIR model with heterogeneous susceptibility and infectivity, allowing for an arbitrary joint distribution of these traits with fixed marginals. Our focus is not on introducing a new epidemic model or re-establishing known qualitative phenomena, but on understanding how epidemic outcomes depend on the dependence structure between susceptibility and infectivity. Using copula theory and the concordance order, we show that both the final epidemic size and the probability of a major outbreak are monotone with respect to this dependence. In particular, among all joint distributions with given marginals, the comonotonic coupling maximizes epidemic outcomes, yielding sharp, distribution-free upper bounds on epidemic risk. This framework provides a unifying analytical explanation for observations previously reported in model-based and numerical studies, including those of \cite{kawagoe2021epidemic} and \cite{tuschhoff2025heterogeneity}, and clarifies which conclusions are robust to distributional assumptions. Copulas play a central role in this analysis, as they allow dependence between susceptibility and infectivity to be characterized independently of their marginal distributions. This is essential because these traits are latent and often defined only up to strictly increasing transformations; unlike correlation, copulas are invariant under such transformations and capture the full dependence structure relevant for distribution-free monotonicity and extremal results.

A central implication of our results is that epidemic outcomes cannot be ordered by susceptibility heterogeneity alone. While increasing susceptibility variance reduces the final size under independence, positive dependence between susceptibility and infectivity can locally increase epidemic size for any basic reproduction number exceeding one. Moreover, neither susceptibility variance nor the sign of the covariance suffices to determine epidemic severity when dependence is present. These findings clarify apparent discrepancies in the literature and explain why conclusions drawn under independence may fail in more realistic settings.


\section{Model and analytical results}
\label{sec:model}


\subsection{Model formulation and the basic reproduction number}
The standard (homogeneous) SIR model relies on the strong assumption that all individuals have identical susceptibility to infection and identical infectiousness once infected. The model is given by 
\begin{align}
S'(t) &=- \beta S(t) I(t),\\
I'(t) &= \beta S(t) I(t) - \gamma I(t),\\
R'(t) &= \gamma I(t),
\end{align}
where $S(t), I(t)$ and $R(t)$ are the fractions of susceptible, infectious, and recovered individuals at time $t$ respectively, $\beta$ is the effective contact rate at which infectious individuals infect susceptible individuals and $\gamma^{-1}$ is the mean infectious period. More general models allowing for individual-level heterogeneity in susceptibility and infectivity have been proposed in the literature; see, for example, \cite{kawagoe2021epidemic}. For convenience, we reconstruct such a model below. Let $(X,C)$ be a pair of nonnegative random variables defined on $\mathbb R_+^2$ with joint distribution $\Pi$ and joint density $\pi(x,c)$, such that $X$ represents the susceptibility factor and $C$ the infectivity factor, with the marginals $\Pi_X$ and $\Pi_C$ respectively. Assume that $\mathbb E[X]=\mathbb E[C]=1$ (otherwise consider rescaled marginals), any scaling can be absorbed into the transmission parameter $\beta$. We also assume that $\mathbb{E}[CX]<\infty$ and $\mathbb{P}(X>0)=\mathbb{P}(C>0)=1.$ In the following, and for $t\ge0$ and $(x,c)\in\mathbb R_+^2$, we
denote by $s(t,x,c),\,i(t,x,c),\,r(t,x,c)$ the densities of susceptible, infectious and recovered individuals with trait $(x,c)$ satisfying
\[
s(t,x,c)+i(t,x,c)+r(t,x,c) = \pi(x,c).
\]
Then the trait-structured SIR model is
\begin{equation}\label{eq:PDEs}
\begin{aligned}
\partial_t s(t,x,c) &= - x\,\Lambda(t)\,s(t,x,c),\\[4pt]
\partial_t i(t,x,c) &= x\,\Lambda(t)\,s(t,x,c) - \gamma\, i(t,x,c),\\[4pt]
\partial_t r(t,x,c) &= \gamma\, i(t,x,c),
\end{aligned}
\end{equation}
with the force of infection
\begin{equation}\label{eq:Lambda}
\Lambda(t) \;=\; \beta \iint_{\mathbb R_+^2} c\, i(t,x,c)\,dx\,dc.
\end{equation}
The total population fractions are given by
\[
S(t)=\iint s(t,x,c)\,dx\,dc,\quad I(t)=\iint i(t,x,c)\,dx\,dc,\quad R(t)=\iint r(t,x,c)\,dx\,dc.
\]

The particular cases where $C=\alpha$ and $C=\alpha\,X$, discussed in (Tekeli 2025 arxiv), are covered by setting $\pi(x,c)=\pi_X(x)\delta(c-\alpha)$ and $\pi(x,c)=\pi_X(x)\delta(c-\alpha\,x)$, respectively, where $\delta$ is the Dirac delta function.

The basic reproduction number, being the average number of new infections produced by a 'typical' index case in a whole susceptible population, depends on what we do mean by a 'typical' infected in the early phase of the epidemic. When the initial infected person is taken uniformly, this number becomes
\begin{align}
R_0=\frac{\beta}{\gamma}\, \mathbb E(C)\, \mathbb E(X)=\frac{\beta}{\gamma},
\end{align}
under the hypothesis $\mathbb E(C)=\mathbb E(X)=1$. However, an individual chosen uniformly at random is not representative of early infectives in a real epidemic. Individuals with higher susceptibility are more likely to be infected first and hence are over-represented among early cases. Following the discussion in Kawagoe et al. (2021), and noticing that the trait distribution among early infectives is proportional to their susceptbility, leading to the following definition of the reproduction number
\begin{equation}\label{eq:R0}
\overline{R}_0 \;=\; \frac{\beta}{\gamma}\,\mathbb E(C X).
\end{equation}
This expression arises naturally as the dominant eigenvalue of the next-generation operator and governs the early exponential growth of infections.

In particular, $\overline{R}_0=R_0$ is the basic reproduction number corresponding to different situations, namely, when there is no heterogeneity, when only one heterogeneity source is considered, and when $X$ and $C$ are independent. When susceptibility and infectivity are dependent, the basic reproduction number  $\overline R_0$ reflects not only average infectiousness but also the composition of early infectives and incorporates the size-biased distribution induced by susceptibility (see Kawagoe et al (2021) for more details). 

In the following we proceed to give the main results related to the final size and major outbreak probability. The proofs are given in the appendix.
\subsection{Final size equation}
We proceed to derive the final size equation under dependent traits. Integrating the first equation in \eqref{eq:PDEs} in time gives
\begin{align}
s(t,x,c) = s(0,x,c)\exp\!\Big(-x\int_0^t \Lambda(u)\,du\Big).
\end{align}
For an outbreak with negligible initial immunity, it is natural to take
$s(0,x,c)\approx\pi(x,c)$, hence
\begin{align}\label{remain_susc}
s(t,x,c) = \pi(x,c)\,e^{-xL(t)},
\end{align}
where $L$ is the cumulative incidence given by \(L(t)=\int_0^t\Lambda(u)\,du\). Thus, at the end of the epidemic the final cumulative incidence is \(L_\infty \;=\; \lim_{t\to\infty} L(t) \).
Integrating the infectious density over time and using \eqref{eq:Lambda} yields the exact scalar fixed-point equation for final incidence $L_\infty$ and thus the final infected fraction (the attack rate) $R_\infty$ as 
\begin{align}
L_\infty &\;=\; \dfrac{\beta}{\gamma}\;
\mathbb E\!\Big[\, C\big(1-e^{-X L_\infty}\big)\Big],\label{eq:final_incid}\\
R_\infty &\;=\; 1 - \mathbb E\!\big[ e^{-X L_\infty} \big].\label{eq:final_attack}
\end{align}
The following result gives the condition for the existence and uniqueness of solution to system \eqref{eq:final_incid}-\eqref{eq:final_attack}.

\begin{theorem}[Existence and uniqueness]\label{thm:existence} The following statements hold for the model \eqref{eq:PDEs}
\begin{enumerate}
\item If $\overline{R}_0\le1$ then the unique solution of \eqref{eq:final_incid} is $L_\infty=0$ and thus $R_\infty=0$.
\item If $\overline{R}_0>1$ then there exists a unique solution $L_\infty>0$ to \eqref{eq:final_incid} and thus $R_\infty>0$.
\end{enumerate}
\end{theorem}

Tuschhoff \& Kennedy (2025) claimed that larger outbreaks can occur for subcritical values of $R_0$. This is not surprising provided that $R_0\le \overline{R}_0$, under positive correlated traits.

\subsection{Concordance order for the final epidemic size}

Generally, the joint traits distribution is rarely unknown. We suggest to use a copula based-approach to establish the monotonicity of the final size in concordance. We first recall the definition of a copula in two-dimensional direction. A copula \(K(u,v)\) is a joint distribution on \([0,1]^2\) with uniform marginals \citep{nelsen2006introduction}. Assuming continuous marginals for $X$ and $C$, Sklar's theorem \citep{nelsen2006introduction} guarantees the existence of a unique copula $K$ verifying
\begin{align}
\Pi(x,c) = K\left( \Pi_X(x),\Pi_C(c) \right),\qquad\forall x,c\in\mathbb{R}.
\end{align}
This implies that the expectation of any transformation $\phi$ with respect to $\Pi$, when it exists, could be expressed as
\begin{align}
\mathbb{E}_{\Pi} \left( \phi(X,Y) \right) = \int\,\phi(x,y)\,dC\left( \Pi_X(x),\Pi_C(c) \right)=\mathbb{E}_K \left( \phi(X,Y) \right).
\end{align}
The association between $X$ and $C$ is commonly expressed in terms of concordance order, capturing more relationships than the correlation coefficient, relying on linear association between two random variables. 
Let $(X_1,C_1)$ and $(X_2,C_2)$ be two bivariate random vectors with joint distributions $\Pi_1$ and $\Pi_2$, respectively, and identical marginals.
We say that $(X_2,C_2)$ is more concordant than $(X_1,C_1)$, and write $(X_1,C_1)\preceq_{c}(X_2,C_2)$ 
if \(\Pi_1(x,c)\le \Pi_2(x,c)\), for all \((x,c)\in\mathbb{R}_+^2\) \citep{joe1990multivariate}.
If $K_1$ and $K_2$ are the corresponding copulas, respectively, this is equivalent to say \citep{cebrian2004testing}, 
\begin{align}
K_1(u,v)\le K_2(u,v)\quad\text{for all }u,v\in[0,1].
\end{align}
Through integration by parts, one can obtain that
\begin{align}\label{order_expect}
\mathbb{E}_{K_1} \left( \phi(X,Y) \right) \le \mathbb{E}_{K_2} \left( \phi(X,Y) \right),
\end{align}
for any integrable map $\phi:\mathbb{R}^2\mapsto\mathbb{R}$ with positive mixed partial derivatives, and the order in \eqref{order_expect} reverses when \(\frac{\partial^2\phi(\cdot,\cdot)}{\partial x \partial y}<0\). Although the associated copula is not unique when the marginal distributions are not continuous, the concordance ordering in \eqref{order_expect} is completely unaffected by copula non-uniqueness.

Now we are able to establish the final size monotonicity result in concordance ordering of traits.
\begin{prop}\label{prop:concord_monoton}
Let $(X_1,C_1)$ and $(X_2,C_2)$ be two trait pairs with identical marginals and such that \((X_1,C_1)\preceq_{c} (X_2,C_2)\). Then, the corresponding basic reproduction numbers satisfy $\overline{R}_0^{(1)}\le \overline{R}_0^{(2)}$. Then the corresponding final attack rates $R_\infty^{(1)}$ and $R_\infty^{(2)}$, respectively, satisfy
\begin{align}
R_\infty^{(1)} \le R_\infty^{(2)}.
\end{align}
\end{prop}
Proposition \ref{prop:concord_monoton} provides an analytical explanation for numerical observations reported by Tushhoff \& Kennedy (2025) stating that correlation of the two heterogeneities tend to increase the final size. Figure \ref{fig:final_size_conc} illustrates Proposition \ref{prop:concord_monoton} under Gaussian copula traits dependence with Gamma marginals. For sufficiently negative dependence, the epidemic remains subcritical and no outbreak occurs as $\overline{R}_0$ remains below 1. Beyond a critical value of the concordance parameter $\rho$, the final size increases in concordance and remains below the final epidemic size corresponding to comonotonic dependent traits.
\begin{figure}[h!]
    \centering
    \includegraphics[width=.7\textwidth]{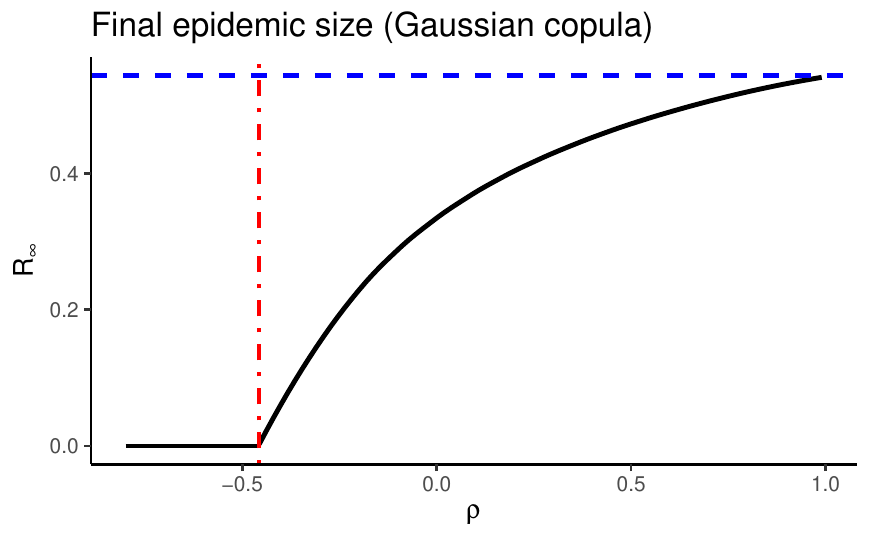}
    \caption{Final epidemic size under a Gaussian copula dependence structure with Gamma marginal distributions, each having mean one and coefficient of variation equal to one, with $R_0=1.5$. The red dashed vertical line indicates the correlation value at which the basic reproduction number $\overline{R}_0$ becomes critical. The blue dashed horizontal line shows the distribution-free upper bound on the final epidemic size, attained under comonotonic dependence.}
    \label{fig:final_size_conc}
\end{figure}

Since any copula is bounded up by the comonotonic copula, i.e. $K(u,v)\le M(u,v) = \min\{u,v\},$ for any $u,v\in[0,1]$. The Proposition \ref{prop:concord_monoton} allows to derive the upper bound of the final size, over all traits distributions, under maximal positive dependence as in the following result.
\begin{cor}\label{prop:final_size_comonotonic}
Assume that $\overline{R}_0>1$. Then, the attack rate satisfies $R_\infty\le R_{M}$, where $R_{M}$ is the solution to
\begin{align}
R_{M} &\;=\; 1 - \mathbb{E}\!\left[e^{-X L_{M}}\right],\label{R_com}\\
L_{M} &\;=\; \dfrac{\beta}{\gamma}\;
\mathbb E_{M}\!\Big[\, C\big(1-e^{-X L_{M}}\big)\Big],\label{L_com}
\end{align}
where $E_{M}$ is the expectation under comonotonic copula $M(u,v) = \min\{u,v\},$ for any $u,v\in[0,1]$.
\end{cor}
We mention that the upper bound $R_{M}$ in Corollary \ref{prop:final_size_comonotonic} is not attainable unless the trait marginals are continuous. Such bound is attainable when $\rho=1$ under Gamma marginals as seen in Figure \ref{fig:final_size_conc}.

\subsection{Non-monotonicity in susceptibility variance}

Under susceptibility-only heterogeneity or independent traits, the final epidemic size is monotone decreasing  in susceptibility variance.  In the following we show that this monotonicity is lost whenever dependence on infectivity is considered. Particularly, we prove the feature, numerically seen by Tuschhoff and Kennedy (2025), that the final size may increase for small variance of susceptibility heterogeneity. In contrast to Tuschhoff and Kennedy (2025), we also prove that this feature does not only hold because of $R_{0}>1$ (nor $\overline{R}_0$) being small or moderate but for any arbitrary $R_0>1$.
\begin{prop}
\label{prop:variance_effect}
Let $(X_\varepsilon,C)$ be a family of trait pairs with
\(
X_\varepsilon = 1 + \varepsilon Z,
\) such that $\mathbb{E}[Z]=0$ and \(\mathbb{P}(X_\varepsilon>0)=1\) for sufficiently small nonnegative $\varepsilon$. Let $R_\infty(\varepsilon)$ be the corresponding final size. Then, the slope $R_\infty'(0)$ is proportional to $\mathbb{E}(CZ)={\rm Cov}(C,X)$.Particularly, the sign of $R_\infty'(0)$ is the same as the sign of ${\rm Cov}(C,X)$.
\end{prop}

Proposition~\ref{prop:variance_effect} shows that the \emph{local} effect of introducing
susceptibility heterogeneity depends on the sign of ${\rm Cov}(C,X)$: when susceptibility and infectivity are positively associated, a small increase in susceptibility variance increases the final epidemic size. Figure \ref{fig:final_size_small} clearly shows local increasing of the final size in susceptibility heterogeneity with an increasing slope in the coefficient of variation of infectivity distribution under comonotonic dependence.
\begin{figure}[h!]
    \centering
    \includegraphics[width=1\textwidth]{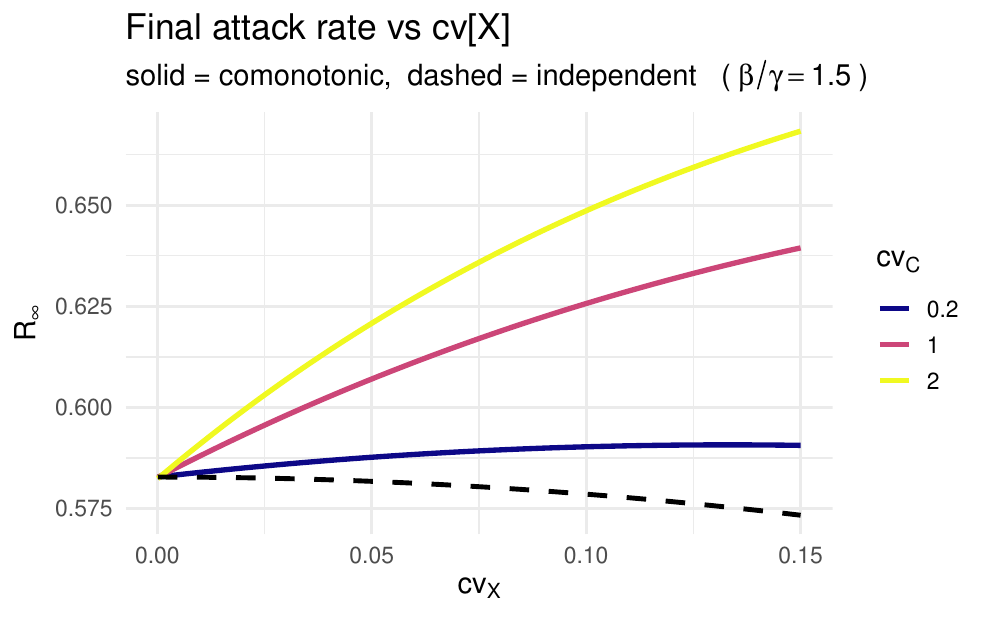}
    \caption{Final epidemic size under comonotonic dependence structure (solid lines) and for independent traits (dashed line). Gamma marginals are used here.}
    \label{fig:final_size_small}
\end{figure}
While it was shown that susceptibility heterogeneity may reduce the final epidemic size under independent traits or when only susceptibility heterogeneity is considered, the following proposition shows that it is not the only reason for epidemic decay/suppression but rather how susceptibility is shaped.

\begin{prop}\label{prop:large_variance}
Let $(X_n,C_n)_{n\ge 1}$ be a sequence of positive susceptibility--infectivity pairs, allowing for arbitrary dependence, such that \(\mathbb{E}[X_n]=\mathbb{E}[C_n]=1\), for all $n\ge 1$, and \(\sup_{n\ge1}\mathbb{E}(X_n C_n)<\infty\). Define for any \(\varepsilon>0\), \( p(\varepsilon) :=\liminf\limits_{n\to \infty} \mathbb{P}(X_n\le\varepsilon)\). If \(\lim\limits_{\varepsilon\to 0}p(\varepsilon) = p\) for some $p\in[0,1]$, then the corresponding final epidemic size $R_\infty^{(n)}$ satisfies
\begin{align}
\limsup_{n\to\infty} R_\infty^{(n)} \;\le\; 1-p .
\end{align}
\end{prop}
The proposition above emphasizes that any increase in the proportion of individuals with near--zero susceptibility necessarily induces a decrease in the final epidemic size, independently of the dependence structure between susceptibility and infectivity. In particular, For Gamma and lognormal families with fixed mean and diverging variance, one has $p(\varepsilon)\to1$ and thus \(R_\infty^{(n)}\to 0\). This is clearly seen in Figure \ref{fig:final_size_dists} for large coefficient of variation. Generally, heavy-tailed distributions such as Pareto laws with fixed mean do not necessarily converge to 0 in probability. In this case, increasing variance is driven by the persistence of a non-negligible upper tail rather than by concentration near zero. As a result, even under extreme heterogeneity, a positive fraction of highly susceptible individuals remains, and the final epidemic size need not vanish (see Figure \ref{fig:final_size_dists}).
\begin{figure}[h!]
    \centering
    \includegraphics[width=1\textwidth]{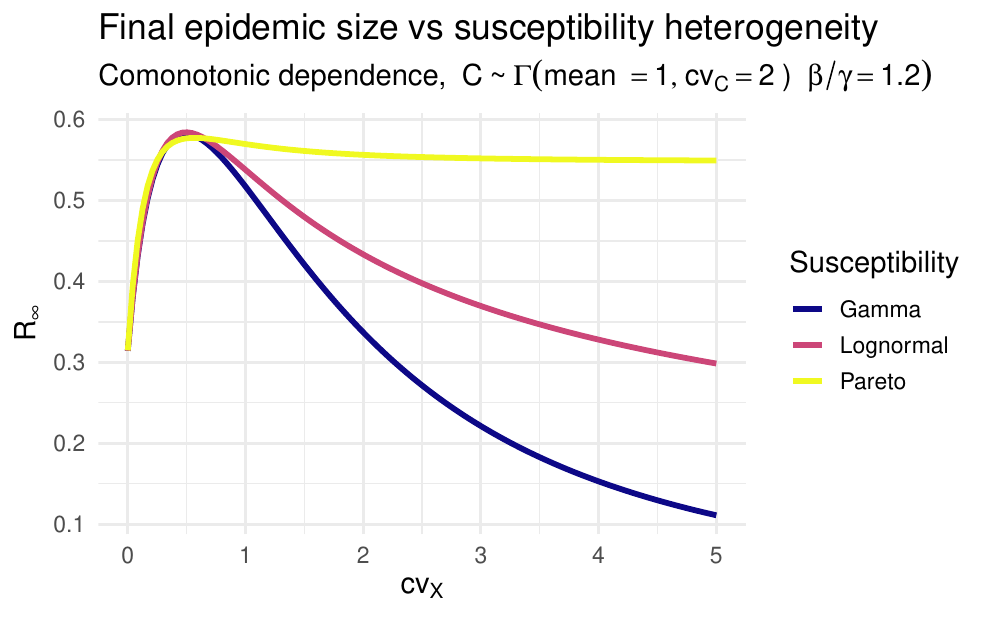}
    \caption{Final epidemic size under comonotonic dependence structure, for different susceptibility distributions with Gamma distributed infectivity.}
    \label{fig:final_size_dists}
\end{figure}

\subsection{Negative correlation does not privilege susceptibility heterogeneity}
To complete the analysis, we also prove through a counterexample that negative covariance of traits distribution does not either ensure monotonicity of the final size. For simplicity, we assume discrete traits distribution.

We define the traits random vector $(X_1,C)$ distribution by 
\begin{align}
(X_1,C)=
\begin{cases}
(0.9,\,1.4),   & \text{with probability } 0.25,\\[4pt]
(0.9,\,0.6),   & \text{with probability } 0.25,\\[4pt]
(1.1,\,1.4),   & \text{with probability } 0.25,\\[4pt]
(1.1,\,0.6),   & \text{with probability } 0.25.
\end{cases}
\end{align}
We also define another traits vector $(X_2,C)$ by
\begin{align}
(X_2,C)=
\begin{cases}
(0.9,\,1.4), & \text{with probability } 0.48,\\[4pt]
(6,\,1.4),   & \text{with probability } 0.02,\\[4pt]
(0.896,\,0.6), & \text{with probability } 0.50,
\end{cases}
\end{align}
so that $\mathbb E(X_1)=\mathbb E(X_2)=\mathbb E(C)=1$, ${\rm Cov}(X_1,C) = -0.04 < 0$, and ${\rm Cov}(X_2,C) = -0.018 < 0$. In this example, we have
\begin{align}
{\rm Var}(X_1)=0.01 < {\rm Var}(X_2)\approx 0.53.
\end{align}
However, a direct numerical solution of the fixed-point equation, for $R_0=1.3$, gives
\begin{align}
R_\infty(X_1,C)\approx 0.36\,<\,
R_\infty(X_2,C)\approx 0.42.
\end{align}
Thus, despite strictly negative susceptibility--infectivity covariance and a small rare mass,
the model with larger susceptibility variance produces a larger final epidemic size. This shows that
negative covariance alone does not guarantee monotonicity of the final size with respect to susceptibility
heterogeneity.

\begin{rek}
While Tekeli (2025) showed that the final size decreases in the susceptibility heterogeneity, say with the parameter $\sigma$(which equals to $\sqrt{1/p}$ therein), when $C$ is proportional to $X$, we emphasize that the factor $\beta/\gamma$ therein is adjusted for any $\sigma$ to maintain a fixed basic reproduction number $\frac{\beta}{\gamma}(1+\sigma^2)$ (see Tekeli 2025), while what is fixed in this paper is $R_{0}=\beta/\gamma$. 
\end{rek}

\subsection{Major outbreak probability}
\label{sec:major_outbreak}

In this section we derive the probability that a single initial infective triggers a major outbreak in the heterogeneous SIR model developed earlier, and show that it is ordered in concordance. To do so, we start by deriving early infections trait distribution.

\subsubsection{Offspring distribution}

Consider the beginning of an epidemic when $s(t,x,c)\approx\pi(x,c)$.  During its infectious period, exponentially distributed with mean $1/\gamma$, an individual with traits $(x,c)$ makes infectious contacts at instantaneous rate $\beta c$ with susceptible individuals randomly chosen from the population. A susceptible with trait $(x,c)$ is infected at rate proportional to $x$. Hence
the density of traits among newly infected individuals is the size-biased
distribution
\(
\pi^{\!*}(x,c)
=x\,\pi(x,c),
\) given $\mathbb E[X]=1$. 
Consequently, a typical infected individual in the branching process has
infectivity trait $C^{\!*}$ with probability density function \(\int \pi^{\!*}(x,c)\,dx\).
Conditional on $C^{\!*}=c$, the number of secondary infections is
$\mathrm{Poi}\!\left(\frac{\beta}{\gamma}c\right)$.
It is often convenient to express expectations with respect to $C^{\!*}$ directly
in terms of the original joint distribution of $(X,C)$ as \(\mathbb E(f(C^{\!*}))
= \mathbb E\!\big( X\,f(C)\big)\)  for any function $f$ for which the expectation exists.
\subsubsection{Major outbreak occurrence and concordance ordering}

The mean number of offspring in the branching process is
\begin{equation}\label{eq:R0_again}
\mathbb E\!\left(\frac{\beta}{\gamma}\,C^{\!*}\right)
= R_0\,\mathbb E(CX) = \overline{R}_0,
\end{equation}
coinciding with the basic reproduction number derived earlier. The generating function of the
offspring distribution can be rewritten as
\begin{equation}\label{eq:G_s_joint}
G(s) = \mathbb E\!\left( X\,\exp\!\left(\frac{\beta}{\gamma}C(s-1)\right)\right).
\end{equation}
The extinction probability $q^\ast$ is the smallest solution in $[0,1]$ of
\begin{equation}\label{eq:extinction_equation}
q = G(q)
= \mathbb E\!\left( X\,\exp\!\left(\frac{\beta}{\gamma}C(q-1)\right) \right),
\end{equation}
and the major outbreak probability is given by $p_{\mathrm{maj}}=1-q^\ast$. 
The following  result shows that $\overline{R}_0$ is the threshold for the major outbreak occurrence.
\begin{prop}\label{major_existence}~~
\begin{enumerate}
\item If $\overline{R}_0\le1$, then the only solution of \eqref{eq:extinction_equation} is $q^\ast=1$ and no major outbreak is possible, that is, $p_{\mathrm{maj}}=0.$
\item If $\overline{R}_0>1$, then the fixed-point equation \eqref{eq:extinction_equation} admits a unique solution $q^\ast<1$,
that is, $p_{\mathrm{maj}}>0.$
\end{enumerate}
\end{prop}

Similar to the final size ordering, the following proposition shows that the major outbreak probability is increasing in concordance order and is maximized under comonotonic dependence. 
\begin{prop}
\label{prop:Pmajor_concord}
For fixed marginals of $(X,C)$, the major outbreak probability is monotone increasing in concordance order of the joint distribution of $(X,C)$. In addition, the largest major outbreak probability yields under the comonotonic traits dependence.
\end{prop}
To illustrate the Proposition \ref{prop:Pmajor_concord}, Figure \ref{fig:maj_prob_conc} plots the the major outbreak probability assuming Gaussian copula dependent traits, where it is clear that major outbreak probability increases in correlation as expected by \cite{tuschhoff2025heterogeneity}.
\begin{figure}[h!]
    \centering
    \includegraphics[width=.7\textwidth]{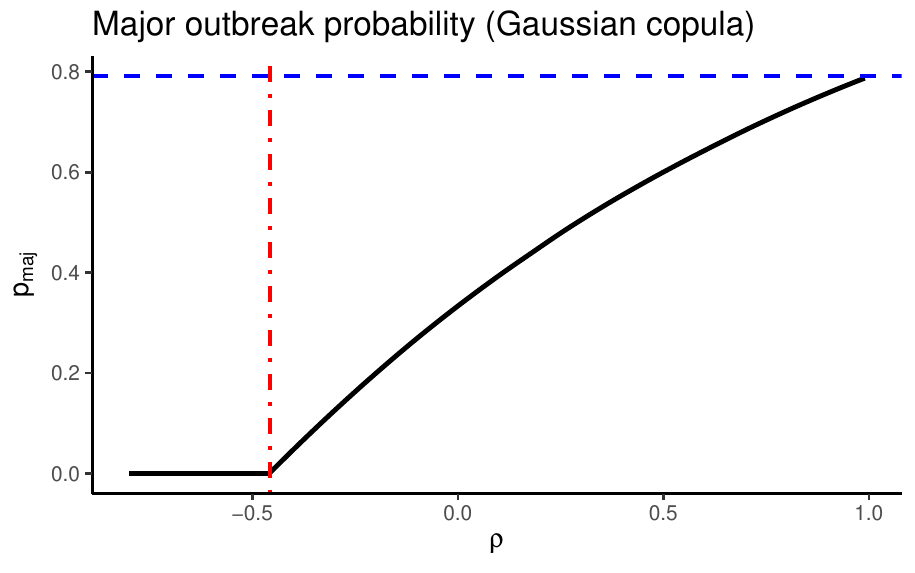}
    \caption{Major outbreak probability under a Gaussian copula dependence structure with Gamma marginal distributions, each having mean one and coefficient of variation equal to one, with $R_0=1.5$. The red dashed vertical line indicates the correlation value at which the basic reproduction number $\overline{R}_0$ becomes critical. The blue dashed horizontal line shows the distribution-free upper bound on the major outbreak probability, attained under comonotonic dependence.}
    \label{fig:maj_prob_conc}
\end{figure}

\section{Discussion}
This paper provides a distribution-free characterization of how dependence between individual susceptibility and infectivity shapes epidemic outcomes in heterogeneous SIR models. By ordering joint trait distributions through the concordance order while holding marginals fixed, we show that both the final epidemic size and the probability of a major outbreak are monotone in dependence, with comonotonic coupling yielding sharp worst-case bounds. These results clarify that epidemic severity cannot, in general, be inferred from susceptibility heterogeneity or covariance alone.

Our analysis explains and extends earlier findings obtained under independence, proportionality, or specific parametric dependence. In particular, while susceptibility heterogeneity is known to reduce final size under independence, we show that positive dependence with infectivity can locally increase epidemic size for any basic reproduction number exceeding one. Conversely, negative covariance does not guarantee suppression. These conclusions resolve apparent contradictions in the literature and demonstrate that dependence structure is a fundamental, and previously under-characterized, driver of epidemic dynamics.
The comonotonic bound derived here provides a rigorous benchmark for epidemic risk assessment when joint trait information is incomplete. Although comonotonicity represents an extreme and idealized scenario, it yields sharp upper envelopes for both early-phase and final-phase epidemic outcomes that are independent of parametric assumptions. This worst-case perspective is particularly relevant in settings where susceptibility is latent and dependence with infectivity is not identifiable from incidence data alone.

Overall, this study highlights that understanding epidemic risk under heterogeneity requires explicit consideration of dependence between traits. Ignoring this structure can lead to qualitatively incorrect conclusions, while bounding arguments based on concordance offer a principled alternative when detailed mechanistic information is unavailable.

Our results are derived for a well-mixed population with permanent immunity and do not directly extend to structured contact networks, time-varying traits, or models with partial immunity or reinfection. Extending concordance-based ordering arguments to such settings remains an open problem. Nonetheless, the present framework isolates dependence effects in their purest form and establishes limits on what can and cannot be inferred from marginal heterogeneity.

\section*{Data availability}
All data generated or analysed during this study are included in this published article and
its supplementary information files.
\bibliographystyle{apalike}
\bibliography{MyBibliography.bib}

@article{ball1993final,
  title={The final size and severity of a generalised stochastic multitype epidemic model},
  author={Ball, Frank and Clancy, Damian},
  journal={Advances in applied probability},
  volume={25},
  number={4},
  pages={721--736},
  year={1993},
  publisher={Cambridge University Press}
}

@article{dwyer1997host,
  title={Host heterogeneity in susceptibility and disease dynamics: tests of a mathematical model},
  author={Dwyer, Greg and Elkinton, Joseph S and Buonaccorsi, John P},
  journal={The American Naturalist},
  volume={150},
  number={6},
  pages={685--707},
  year={1997},
  publisher={The University of Chicago Press}
}

@article{gomes2022individual,
  title={Individual variation in susceptibility or exposure to SARS-CoV-2 lowers the herd immunity threshold},
  author={Gomes, M Gabriela M and Ferreira, Marcelo U and Corder, Rodrigo M and King, Jessica G and Souto-Maior, Caetano and Penha-Gon{\c{c}}alves, Carlos and Gon{\c{c}}alves, Guilherme and Chikina, Maria and Pegden, Wesley and Aguas, Ricardo},
  journal={Journal of theoretical biology},
  volume={540},
  pages={111063},
  year={2022},
  publisher={Elsevier}
}

@article{novozhilov2008spread,
  title={On the spread of epidemics in a closed heterogeneous population},
  author={Novozhilov, Artem S},
  journal={Mathematical biosciences},
  volume={215},
  number={2},
  pages={177--185},
  year={2008},
  publisher={Elsevier}
}

@article{kawagoe2021epidemic,
  title={Epidemic dynamics in inhomogeneous populations and the role of superspreaders},
  author={Kawagoe, K and Rychnovsky, M and Chang, S and Huber, G and Li, LM and Miller, J and Pnini, R and Veytsman, B and Yllanes, D},
  journal={Physical Review Research},
  volume={3},
  number={3},
  pages={033283},
  year={2021},
  publisher={APS}
}

@article{tuschhoff2025heterogeneity,
  title={Heterogeneity in and correlation between host transmissibility and susceptibility can greatly impact epidemic dynamics},
  author={Tuschhoff, Beth M and Kennedy, David A},
  journal={Journal of Theoretical Biology},
  pages={112186},
  year={2025},
  publisher={Elsevier}
}

@book{nelsen2006introduction,
  title={An introduction to copulas},
  author={Nelsen, Roger B},
  year={2006},
  publisher={Springer}
}

@article{britton2020mathematical,
  title={A mathematical model reveals the influence of population heterogeneity on herd immunity to SARS-CoV-2},
  author={Britton, Tom and Ball, Frank and Trapman, Pieter},
  journal={science},
  volume={369},
  number={6505},
  pages={846--849},
  year={2020},
  publisher={American Association for the Advancement of Science}
}

@article{lloyd2005superspreading,
  title={Superspreading and the effect of individual variation on disease emergence},
  author={Lloyd-Smith, James O and Schreiber, Sebastian J and Kopp, P Ekkehard and Getz, Wayne M},
  journal={Nature},
  volume={438},
  number={7066},
  pages={355--359},
  year={2005},
  publisher={Nature Publishing Group UK London}
}

@article{tkachenko2021time,
  title={Time-dependent heterogeneity leads to transient suppression of the COVID-19 epidemic, not herd immunity},
  author={Tkachenko, Alexei V and Maslov, Sergei and Elbanna, Ahmed and Wong, George N and Weiner, Zachary J and Goldenfeld, Nigel},
  journal={Proceedings of the National Academy of Sciences},
  volume={118},
  number={17},
  pages={e2015972118},
  year={2021},
  publisher={National Academy of Sciences}
}

@article{joe1990multivariate,
  title={Multivariate concordance},
  author={Joe, Harry},
  journal={Journal of multivariate analysis},
  volume={35},
  number={1},
  pages={12--30},
  year={1990},
  publisher={Elsevier}
}

@article{cebrian2004testing,
  title={Testing for concordance ordering},
  author={Cebri{\'a}n, Ana C and Denuit, Michel and Scaillet, Olivier},
  journal={ASTIN Bulletin: The Journal of the IAA},
  volume={34},
  number={1},
  pages={151--173},
  year={2004},
  publisher={Cambridge University Press}
}

\newpage
\appendix
\section*{Appendix}

\begin{proof}[Proof of Theorem~\ref{thm:existence}]
Define the map
\[
F(L) := \dfrac{\beta}{\gamma}\,\mathbb E\!\big[ C(1-e^{-XL})\big],\qquad L\ge0.
\]
Then $F$ is continuous and strictly increasing on $[0,\infty)$, with $F(0)=0$. In addition, for any $x,c\ge 0$, the function \(L\mapsto c\left( 1-e^{-xL}\right) \) is concave on $\mathbb{R}_+$. By linearity of the expectation, we obtain that $F$ is concave on $\mathbb{R}_+$.
The tangent at $0$ has the slope 
\begin{align}
F'(0)=  \dfrac{\beta}{\gamma}\mathbb E(CX)=\overline{R}_0.
\end{align}
Then, if $\overline{R}_0\le 1$, the unique solution to \eqref{eq:final_incid} is $L_\infty=0$ and thus the corresponding $R_\infty = 0$. Otherwise, there is a unique positive solution $L_\infty>0$ to \eqref{eq:final_incid} which results in a unique positive final attack rate $R_\infty>0$.
\end{proof}

\begin{proof}[Proof of Proposition \ref{prop:concord_monoton}]
We recall that the final--size equation satisfied by the couple $(X_j,C_j)$, $j=1,2$, is given by
\begin{equation}\label{eq:final_size_repeat}
L_\infty^{(j)} \;=\; \frac{\beta}{\gamma}\,
\mathbb{E}\!\left[\, C_j\bigl(1-e^{-X_j L_\infty^{(j)}}\bigr) \right],
\end{equation}
whose unique positive solution determines the final fraction
\[
R_\infty^{(j)} = 1 - \mathbb{E}\!\left[e^{-X_j L_\infty^{(j)}}\right].
\]
The corresponding reproduction number is given by \(\overline{R}_0^{(j)} = R_{0}\mathbb{E}(X_jC_j)\) for $j=1,2$. Since the function \(\phi(x,c)=xc\) satisfies \(\frac{\partial^2\phi(x,c)}{\partial x \partial y}=1>0\), it follows from \eqref{order_expect} that $\overline{R}_0^{(1)}\le \overline{R}_0^{(2)}$.

Assume now that $\overline{R}_0^{(1)}>1$ so both final-size equations admit unique positive solutions \(L_\infty^{(1)}>0\) and \(L_\infty^{(2)}>0\) (otherwise, the inequality follows immediately). Define for $j=1,2$,
\begin{align*}
F_j(L) := \dfrac{\beta}{\gamma}\,\mathbb E\!\big[ C_j(1-e^{-X_jL})\big],\qquad L\ge0.
\end{align*}
Consider the function \(
\varphi_L(x,c) := c\bigl(1-e^{-xL}\bigr).
\)
Then the mixed partial derivative satisfies
\[
\frac{\partial^2}{\partial x\,\partial c}\varphi_L(x,c)
= Le^{-xL} > 0.
\]
Hence, by the concordance order, we obtain for all $L\ge 0$
\[
F_1(L)\;\le\;
F_2(L),
\]
which implies that \(L_\infty^{(1)} \le L_\infty^{(2)}\).
Since $X_1$ and $X_2$ have the same marginals, we obtain
\begin{align*}
R_\infty^{(1)} \le R_\infty^{(2)}.
\end{align*}
\end{proof}

\begin{proof}[Proof of Proposition~\ref{prop:variance_effect}]

Let $L_\infty(\varepsilon)$ denote the solution of the final--size equation
\begin{equation}\label{eq:final_size_eps}
L_\infty(\varepsilon)
= \frac{\beta}{\gamma}\,
\mathbb{E}\!\left[C\bigl(1-e^{-L_\infty(\varepsilon) X_\varepsilon}\bigr)\right],
\end{equation}
and let $R_\infty(\varepsilon)=1-\mathbb{E}\!\left[e^{-X_\varepsilon L_\infty(\varepsilon)}\right]$ be the corresponding attack rate. Differentiating and taking the expectation yields
\begin{equation}\label{eq:Rprime_formula}
R_\infty'(0)
=
e^{-L_\infty(0)} L_\infty'(0).
\end{equation}
with $L_\infty(0)>0$ provided that $R_0= \beta/\gamma\, > 1$, corresponding to the homogeneous case $X=1$.
Define
\[
F(L,\varepsilon)
:= L - \frac{\beta}{\gamma}\,
\mathbb{E}\!\left[C\bigl(1-e^{-LX_\varepsilon}\bigr)\right].
\]
Then $L_\infty(\varepsilon)$ is characterized implicitly by
\(
F(L_\infty(\varepsilon),\varepsilon)=0.
\)
To compute $L_\infty'(0)$, we use the implicit function theorem 
\begin{equation}\label{eq:Lprime_general}
L_\infty'(0)
=
-\frac{\partial_\varepsilon F(L_\infty(0),0)}{\partial_L F(L_\infty(0),0)}.
\end{equation}
From $X_\varepsilon = 1 + \varepsilon Z$, we obtain
\(
\partial_\varepsilon e^{-L_\infty(0) X_\varepsilon}\big|_{\varepsilon=0}
= -L_\infty(0) Z e^{-L_\infty(0)},
\)
which implies that
\begin{align*}
\partial_\varepsilon F(L_\infty(0),0)
= 
\frac{\beta}{\gamma}\,
\mathbb{E}\!\left[
C\bigl(\partial_\varepsilon e^{-L_\infty(0) X_\varepsilon}\bigr)\big|_{\varepsilon=0}
\right]
= -
\frac{\beta}{\gamma}\,L_\infty(0) e^{-L_\infty(0)}\, \mathbb{E}(C Z).
\end{align*}
In addition, we have
\[
\partial_L F(L,0)
=
1 - \frac{\beta}{\gamma}\,\mathbb{E}\!\left[C e^{-L}\right].
\]
Then, evaluating at $L=L_\infty(0)$, we obtain
\begin{equation}\label{eq:denominator_negative}
\partial_L F(L_\infty(0),0)
= 1 - \frac{\beta}{\gamma}e^{-L_\infty(0)}.
\end{equation}
Since $L_\infty(0)=\frac{\beta}{\gamma}\left( 1-e^{-L_\infty(0)}\right) $, one can obtain that
\begin{align*}
1 - \frac{\beta}{\gamma}e^{-L_\infty(0)} = 1 - \frac{L_\infty(0)e^{-L_\infty(0)}}{1-e^{-L_\infty(0)}} >0.
\end{align*}
Thus, from equation \eqref{eq:Lprime_general}, $L_\infty'(0)$ is proportional to \(\mathbb{E}(CZ)={\rm Cov}(X,C)\) and thus is $R_\infty'(0)$. Particularly,
\[
\mathrm{sign}\bigl(R_\infty'(0)\bigr)
=\mathrm{sign}\left( {\rm Cov}(X,C)\right) .
\]
\end{proof}

\begin{proof}[Proof of Proposition \ref{prop:large_variance}]
Recall that the final size satisfies
\[
R_\infty^{(n)} = 1-\mathbb{E}\!\left(e^{-L_n X_n}\right),
\]
where $L_n$ is the unique solution of
\[
L_n = \frac{\beta}{\gamma}\,
\mathbb{E}\!\left[C_n\bigl(1-e^{-X_n L_n}\bigr)\right].
\]
Using the inequality $0\le 1-e^{-y}\le y$ for $y\ge0$ and the uniform boundedness of $\mathbb{E}(X_n C_n)$, we obtain
\[
L_n \le \frac{\beta}{\gamma}\,\mathbb{E}[X_n C_n] \le M,
\]
for some constant $M>0$. Hence
\begin{align*}
\mathbb{E}\!\left(e^{-L_n X_n}\right)
&\;\ge\;
\mathbb{E}\!\left(e^{-M X_n}\right)\\
&\;\ge\; \mathbb{E}\!\left(e^{-M X_n}\mathbf{1}_{\{X_n\le\varepsilon\}}\right)\\
&
\;\ge\;
e^{-M\varepsilon}\,\mathbb{P}(X_n\le\varepsilon),
\end{align*}
for any $\varepsilon>0$.
Letting $n\to\infty$ and using the assumption that \( p(\varepsilon)\to p\in[0,1]\) as $\varepsilon\to 0$, we obtain
\[
\liminf_{n\to\infty}\mathbb{E}\!\left(e^{-L_n X_n}\right)\ge p.
\]
Therefore,
\[
\limsup_{n\to\infty} R_\infty^{(n)}
=
1-\liminf_{n\to\infty}\mathbb{E}\!\left(e^{-L_n X_n}\right)
\;\le\; 1-p,
\]
which completes the proof.
\end{proof}

\begin{proof}[Proof of Theorem \ref{major_existence}]
The proof is similar to the proof of Theorem \ref{thm:existence} so we omit it.
\end{proof}

\begin{proof}[Proof of Proposition \ref{prop:Pmajor_concord}]
Fix $s\in(0,1)$ and define the map
\begin{align*}
f_s(x,c) := x\,e^{\frac{\beta}{\gamma}(s-1)\, c},\quad\forall\, x,c\ge 0.
\end{align*}
Then from \eqref{eq:G_s_joint} we have
\[
G(s) = \mathbb E\!\big(f_{s}(X,C)\big).
\]
Note that
\[
\frac{\partial^2}{\partial x\,\partial c} f_s(x,c)
= \frac{\beta}{\gamma}(s-1) e^{\frac{\beta}{\gamma}(s-1)\, c} < 0,\quad\forall s\in(0,1).
\]
Let $(X_1,C_1)$ and $(X_2,C_2)$ be two bivariate random vectors coupled via copulas $K_1$ and $K_2$, respectively, with the same marginals. Assume that $(X_1,C_1)\preceq_{c}(X_2,C_2)$. That is, $K_1\le K_2$ on $[0,1]^2$. Then the corresponding probability generating functions $G^{(1)}$ and $G^{(2)}$, respectively, satisfy
\begin{align*}
G^{(1)}(s) \ge G^{(2)}(s),\quad\forall\, s\in(0,1).
\end{align*}
Then the corresponding fixed points $q^{(1)}$ and $q^{(2)}$ (the extinction probabilities) respectively, satisfy $q^{(1)} \ge q^{(2)}$. That is, the major outbreak probabilities are ordered as
\begin{align}
p_{\rm maj}^{(1)} = 1-q^{(1)} \le 1-q^{(2)} = p_{\rm maj}^{(2)}.
\end{align}
The maximal major outbreak probability is achieved when $K_2=M$, on $[0,1]^2$ with $M$ is the comonotonic copula. This completes the proof.
\end{proof}

\end{document}